\documentclass[a4paper,12pt]{article}
\makeatletter
\newcommand*{\rom}[1]{\expandafter\@slowromancap\romannumeral #1@}
\makeatother
\usepackage{etex}
\usepackage{amsmath,amsthm,amsfonts,amssymb,tikz-cd,mathdesign,float,mathtools}
\usepackage{lmodern}
\usepackage{relsize}
\usepackage{euler}
\usepackage{listings}

\definecolor{codegreen}{rgb}{0,0.6,0}
\definecolor{codegray}{rgb}{0.5,0.5,0.5}
\definecolor{codepurple}{rgb}{0.58,0,0.82}
\definecolor{backcolour}{rgb}{0.95,0.95,0.92}

\lstdefinestyle{mystyle}{  
	commentstyle=\color{codegreen},
	keywordstyle=\color{blue},
	numberstyle=\tiny\color{codegray},
	stringstyle=\color{codepurple},
	basicstyle=\ttfamily\footnotesize,
	breakatwhitespace=false,         
	breaklines=true,                 
	captionpos=b,                    
	keepspaces=true,                 
	numbers=left,                    
	numbersep=5pt,                  
	showspaces=false,                
	showstringspaces=false,
	showtabs=false,                  
	tabsize=2
}

\usepackage{times}
\usetikzlibrary{positioning,automata}
\usepackage{mathtools}
\usepackage{youngtab}
\usepackage{pgfplots}
\usepackage{pst-plot}
\usepackage{pgfplots}
\usepackage[alphabetic]{amsrefs}
\pgfplotsset{compat=1.10}
\usepgfplotslibrary{fillbetween}
\usetikzlibrary{patterns}
\usetikzlibrary{matrix}
\numberwithin{equation}{section}

\newtheorem{thm}{Theorem}[section]

\newtheorem{cor}[thm]{Corollary}

\theoremstyle{definition}
\newtheorem{defn}[thm]{Definition}
\newtheorem{rmk}[thm]{Remark}
\newtheorem{ex}[thm]{Example}

\usepackage[pdftex,colorlinks=true,allcolors=blue]{hyperref}

\DeclareMathAlphabet\mathbfcal{OMS}{cmsy}{b}{n} 

\usepackage{xcolor}

\usepackage{multirow}
\usepackage{makeidx}
\usepackage{latexsym}
\usepackage{amsfonts}
\usepackage{amssymb}
\usepackage{amsmath}
\usepackage{amstext}
\usepackage{amsthm}
\usepackage{mathrsfs}
\usepackage{float}

\makeatletter
\renewcommand{\BibLabel}{ \Hy@raisedlink{\hyper@anchorstart{cite.\CurrentBib}\hyper@anchorend}}
\makeatother

\begin{document}
	
\title{Manipulable outcomes within the class of scoring voting rules\thanks{The authors would like to thank Eric Kamwa for comments and helpful discussions. The first author gratefully acknowledges the financial supports from Universit\'e de Lyon (project INDEPTH Scientific Breakthrough Program of IDEX Lyon) within the program Investissement d'Avenir (ANR-16-IDEX-0005) and from Universit\'e de Franche-Comt\'e within the program Chrysalide-2020. The second author wishes to express his gratitude to F. Aleskerov for introducing him to the subject and constant guidance on the early stages of development of this project. He also  acknowledges the support of the International Laboratory of Decision Choice and Analysis (National Research University Higher School of Economics).}}
	
	\author{Mostapha Diss\thanks{CRESE EA3190, Univ. Bourgogne Franche-Comt\'e, F-25000 Besan\c con, France. Email: mostapha.diss@univ-fcomte.fr.} \and Boris Tsvelikhovskiy\thanks{Department of Mathematics, Northeastern University, 	Boston, MA, 02115, USA. Email:  tsvelikhovskiy.b@northeastern.edu.}}
	
	\maketitle

\begin{abstract}Coalitional manipulation in voting is considered to be any scenario in which  a group of voters decide to misrepresent their vote in order to secure an outcome they all prefer to the first outcome of the election when they vote honestly. The present paper is devoted to study coalitional  manipulability within the class of scoring voting rules.  For any such rule and any number of alternatives,  we introduce a new approach allowing  to characterize all the outcomes that can be manipulable by a coalition of voters. This gives us the possibility to find the probability of manipulable outcomes for some well-studied scoring voting rules in case of small number of alternatives  and  large electorates under a well-known assumption on individual preference profiles.
\end{abstract}

\section{Introduction}\label{intro}

Since the seminal papers of  \cite{Gibbard}  and  \cite{Satterthwaite}  who proved that  every  non-dictatorial  social choice rule  can be manipulated in the presence of at least three alternatives, the problem of coalitional manipulation has received a lot of attention in recent decades in social choice theory. Broadly speaking,  a given social choice rule is called  \emph{coalitionally manipulable}  if there exists a given list of voting preferences and a coalition  of voters,\footnote{A coalition is defined as any non-empty subset of voters.}   such that the preferences of all the voters outside the coalition remain the same, while the preferences of voters within the coalition can be altered in such a way that the winner changes and each of the voters from the coalition is `happy about the change'. 

Scoring voting rules, also called \emph{positional voting rules}, have attracted a considerable amount of attention in the literature dealing with manipulation. This class of voting rules can be defined as follows: each voter's preference must be a vector that gives a number of points that the voter assigns to each alternative according to his or her position in the voter's preference. The points assigned  by  all voters are summed and the winning alternative has the highest number of points.  A number of studies has been conducted on the evaluation of the degree of manipulability of various social choice rules, i.e., the extent to which social choice rules are manipulable by a coalition of voters or by an individual voter. The reader may refer, for instance, to \cite{AlKurb,Chamberlin,D15,Ouafdi1,Ouafdi2,FavardinLepelley,Favardin,GehrleinMoyouwouLepelley,KaM,Kelly,KimRoush,Lepelley,Lepelley2,MoyouwouT2017,Nitzan,Peleg,Wilson1,Wilson2,Saari90}, and \cite{Schurmann}.  The methodology used in this literature consists first in characterizing the specific conditions that must be required for a given voting rule to be manipulable by a coalition of voters. The final step consists in the evaluation of the (theoretical) probability of this situation under various assumptions on voters' preferences. For more details on those probabilistic assumptions and their use in social choice theory, the reader can refer to  \cite{DM20,GehrleinLepelley2,GehrleinLepelley1}.

However, as pointed out in a recent paper  by  \cite{Ouafdi1}, with some notable exceptions the results appearing in the literature only deal with three-alternative elections, not because it is the most interesting case but due to the difficulties arising when considering more than three alternatives.  The main goal of this paper is to provide a significant improvement in this direction.  This is achieved via presenting systems of linear inequalities, which determine manipulable profiles, that is the lists defining the vote of all voters taking part in the decision process,  for the whole class of scoring voting rules independently on the number of alternatives.  More precisely, we focus on a detailed exposition of the new approach for obtaining the list of linear inequalities a profile satisfies if and only if it is manipulable in case of $m\geq 3$ alternatives for all scoring rules.  

The paper is organized as follows. Section \ref{prelimin} describes the basic framework.  The core of the paper is Section \ref{mainsection}, where the main results are presented. We start of with providing a simplified illustration of how our methodology works in case of three-alternative elections. Then, the precise list of inequalities for the whole class of scoring rules is presented.  Attention will be focused on the limiting case where the number of voters tends to infinity.  This will enable us to complete the existing literature dealing with the probability of manipulable outcomes by providing in  Section \ref{sectionproba}  the corresponding values for the three most studied scoring rules in this literature which are Plurality, Antiplurality and Borda rules in the presence of $m=4$ and $m=5$  alternatives under the well-known Impartial and Anonymous Culture assumption (defined later). The results are approximate since they are obtained by Monte-Carlo simulations, but with a high degree of precision, since we used profiles of cardinality $8\cdot 10^6$   for $m=4$ and $8\cdot 10^5$ for $m=5$ alternatives.  To the best of our knowledge none of those results  have appeared in the literature, with the exception of the Plurality rule  and  $m=4$ (see \cite{Ouafdi1}).   The new results clearly indicate that among the three positional rules under consideration Antiplurality has a much lesser degree of manipulability (see Tables \ref{3CandTable},  \ref{4CandTable}, and \ref{5CandTable}). The last section presents our conclusions.

\section{Preliminary definitions}\label{prelimin}

Let us assume that the number of voters/individuals is denoted by $n$ and the number of alternatives/candidates   by  $m$.  The voters' preferences are assumed to be linear orders which means that voters rank alternatives from
most preferred to least preferred and indifference is not allowed.  
An  \emph{anonymous  preference profile}, hereafter simply called a \emph{profile},  is an $m!$-tuple of non-negative integer numbers $(n_{1}, n_{2},\ldots, n_{m!})$  such that $n_{i}\geq 0$ for all $m! \geq i \geq 1$ and $\sum\limits_{i=1}^{m!} n_{i}=n$.\footnote{This is also known as a \emph{voting situation} in the literature.} Each $n_{i}$ is equal to the number of voters with preferences of type $i$.  
In the limit as $n\to \infty$,  we consider the normalized profile vectors of the form $p=(p_1,\hdots,p_{m!})$ with each $p_i=\frac{n_i}{n}\geq 0$ and $\sum\limits_{i=1}^{m!}p_i= 1$.  
Throughout the paper, the alternatives will be called $A_1$, $A_2$, \dots, $A_m$ and  the types of possible preference rankings of the alternatives will always be listed following an increasing order of $j$ in the notation $A_j$. In addition, the proportion of individuals having each possible preference ranking will be denoted accordingly. For instance, in case of $m=3$ alternatives, the  possible preference rankings are denoted by  $(A_1,A_2,A_3)$ $p_{1}$, $(A_1,A_3,A_2)$  $p_{2}$, $(A_2,A_1,A_3)$  $p_{3}$, $(A_2,A_3,A_1)$  $p_{4}$, $(A_3,A_1,A_2)$  $p_{5}$ and $(A_3,A_2,A_1)$  $p_{6}$. The notation $(A_1,A_2,A_3)$ $p_{1}$, for instance, means that a fraction $p_1$ of individuals  have preferences  with $A_1$  being most preferred, $A_3$ being least preferred and with $A_2$ being ranked between them.

In the case of  $m$ alternatives and individual preferences expressed   as  linear  orders,  we  can define the class of scoring voting rules as follows:

\begin{defn}
In the case of linear orders on $m$ alternatives, every scoring voting rule can be defined by the  \textit{weight vector}   $w=(w_1,w_2,\hdots,w_m)$  such  that  $w_1\geq w_2\geq\hdots\geq w_m$  and  $w_1>w_m$.   It means that  each time an alternative is ranked $r$-th by one voter it obtains $w_r$ points;  afterwards we select the alternative(s)  obtaining  the greatest aggregated score.
\end{defn}

Given this definition, we are now ready to describe three well-known scoring voting rules. Under Plurality rule each voter has one vote which he/she can cast for any one of the $m$ alternatives, i.e.,  $w=(1,0,\dots, 0)$.  Under Borda rule, we select the  alternative with the highest Borda score such that  each first-place vote is worth $m-1$ points, each second-place vote is worth $m-2$ points, and so on until $0$ point to each last-place vote, i.e.,   $w=(m-1,m-2,\dots, 1, 0)$. Under Antiplurality rule, we return the  alternative with the  highest aggregated score such that each voter assigns one point for any one of the $m-1$  best ranked alternatives, i.e.,  $w=(1,1,\dots,1, 0)$.  

Note that in the presence of three alternatives  $A_1$, $A_2$, and $A_3$, we see that any  general scoring rule $w=(w_1,w_2,w_3)$  is equivalent to the normalized scoring vector $(1,\lambda, 0)$  with  $0\leq \lambda\leq 1$. This is obtained by subtracting $w_3$ and dividing by $w_1-w_3$ any component of $w$. The well-known scoring rules become the Plurality rule with $\lambda=0$, the Borda rule with $\lambda=0.5$, and the Antiplurality rule with $\lambda=1$.

As already mentioned we focus in this paper on the coalitional manipulation which can be defined as follows:

\begin{defn}
The outcome of an election is said to be  \emph{coalitionally manipulable} if there exists an alternative  such that all members of the electorate for whom she is preferable over the winning alternative can change their preferences in such a way that this alternative wins the election. The preferences of the rest of the electorate remain the same. 
\end{defn}

We mainly focus in this paper on the manipulation by coalitions of maximal sizes but we will also show that  the approach that we consider  can be adapted when  manipulation by coalitions of smaller sizes of the electorate is
considered. 

We say that the initial arrangement  is  $(A_1,A_2,\hdots,A_m)$ if the collective ranking of the alternatives before manipulation is: $A_1$ first, $A_2$ second and so on until alternative  $A_m$ which is ranked last.  Let us consider an example in order to highlight the notion of manipulation.

\begin{ex} \label{BordaManipEx} Let $m=3$ and the scoring rule be Borda rule, i.e.,  $w=(2,1,0)$. To show a manipulable outcome, we set  $p_{1}=5/9$,  $p_{3}=4/9$, and  $p_{j}=0$  for $j\neq 1,3$. By the definition of Borda rule, the initial arrangement is  $(A_1,A_2,A_3)$ since alternative $A_1$ wins and gets $14/9$ points, $A_2$ gets $13/9$ points,  and finally $A_3$ gets $0$ points.  But, if  voters having the preference $(A_2,A_3,A_1)$ chose to vote  $(A_2,A_1,A_3)$ instead, then the final result would be in favor of alternative $A_2$ since $A_1$ gets $10/9$ points, $A_2$ gets $13/9$ points, and $A_3$ gets $4/9$ points.      \end{ex}

With a given initial arrangement,   $m-1$  different coalitions can be formed. This is due to the fact that for each alternative, who has not won the election, there may be a group of voters who rank it higher than the winner. As any coalition is determined by the alternative it would prefer to have as the winner of the election, the alternatives' names will be used for the notations of the coalitions hereafter. Alternative ${A_i}$ will be called \emph{unifying} for the coalition denoted by $Coal_{A_i}$ that would prefer to have ${A_i}$ as the winner and let us denote by $\mathfrak{C}_{A_i}$  the proportion of individuals in this coalition. In the case of $m=3$ alternatives, for instance, two coalitions can be formed  since the unifying alternative can be any alternative except the winner.  In this case, when the winner is alternative $A_1$, we have that alternative $A_2$ is unifying for the coalition $Coal_{A_2}$ formed by voters with preferences $(A_2,A_1,A_3),(A_2,A_3,A_1)$ and $(A_3,A_2,A_1)$  and alternative $A_3$ is unifying for the coalition $Coal_{A_3}$ formed by voters with  preferences  $(A_3,A_1,A_2),(A_3,A_2,A_1)$ and $(A_2,A_3,A_1)$. It is clear that the preferences  $(A_2,A_3,A_1)$ and $(A_3,A_2,A_1)$ participate in both coalitions $Coal_{A_2}$ and $Coal_{A_3}$.

We introduce now the concept of  \textit{intermediate preferences}.
 
 \begin{defn} \textit{Intermediate preferences} are  preferences of the form $(*,*,\hdots,*)$ with each "$*$" symbol being either an alternative's name or the "$?$" mark, the latter representing the fact that its corresponding rank has yet to be assigned by the voter.  \end{defn}
 
Typical examples of  intermediate preferences  in three-alternative elections are $(?,?,?)$  where the concerned  voter is undecided on all of the three positions and $(A_1,?,?)$ where the concerned voter ranks alternative $A_1$ at the first position and he/she is undecided on the second and third positions in his/her preference.

Let us now introduce the number $d(A_k,A_i)$, with $i \neq k$,  which denotes the difference in points between the unifying alternative $A_k$  and alternative  $A_{i}$ prior to the coalition participants' arrangement of all the places except the first.  In the case of three alternatives, when the winner is alternative $A_1$ for instance, individuals of $Coal_{A_2}$ (of proportion $\mathfrak{C}_{A_2}=p_{3}+p_{4}+p_{6}$)  vote for alternative $A_2$, i.e., the unifying alternative of that coalition. In addition, those individuals only assign  weight  $1$ to their first preferred alternative and zero weight to the other alternatives prior to the coalition participants' arrangement of the second and third places. In other words, those individuals have an intermediate preference of type $(A_2,?,?)$. Note that the other individuals (of proportion $p_{1}+p_{2}+p_{5}$)  assign all their points and vote sincerely when we focus on the manipulation in favor of $A_2$.  As a consequence, the score of alternative $A_2$ is $\lambda p_1+ p_3+p_4+p_6$ and the score of alternative $A_1$ is $p_1+p_2+\lambda p_5$.  This leads to $d(A_2,A_1)=p_{3}+p_{4}+p_{6}-p_2+(\lambda -1)p_1-p_{5}$.

\section{Conditions describing manipulability}\label{mainsection}

In this section we show that for every coalition there exists a system of inequalities which characterize the set of all the profiles that are manipulable by it. In order to make a  simplified  illustration  of the general case that we will provide in Theorem \ref{mainThm},  let us start first by the case of  three-alternative elections. Let again the alternatives be referred to as $A_1$, $A_2$, and $A_3$ and consider the representation $w=(1,\lambda, 0)$, with $0\leq \lambda\leq 1$, that defines all possible scoring rules with three alternatives. Using all the above definitions, we are ready to provide the list of inequalities which characterizes the set of all the profiles $p=(p_1,\hdots,p_6)$  that are manipulable by a coalition of voters for  three alternatives.  The first system indicates conditions where manipulation is in favour of  $A_2$ whereas the second system is given for a manipulation in favour of  $A_3$.
\begin{equation}
\begin{split}\label{ManipEq}
\begin{cases}
 p_i\geq 0, ~i\in\{1,2,\hdots,6\}\\
  p_{1}+p_{2}+p_{3}+p_{4}+p_{5}+p_{6}=1 \\
 (1-\lambda)p_{1}+p_{2}+(\lambda-1)p_{3}- p_{4}+\lambda p_{5}-\lambda p_{6}>0, & A_1~beats~A_2  \\ \lambda p_{1}-\lambda p_{2}+p_{3}+(1-\lambda)p_{4}- p_{5}+(\lambda-1)p_{6}>0, & A_2~beats~A_3  \\  p_{3}+p_{4}+p_{6}-p_2+(\lambda-1)p_1-\lambda p_{5}>0, & d(A_2,A_1)>0 \\ (2-\lambda)(p_{3}+p_{4}+p_{6})+(2\lambda-1)p_1-(1+\lambda)p_2-(1+\lambda)p_5>0, & d(A_2,A_1)+d(A_2,A_3)>\lambda\mathfrak{C}_{A_2}\\  \end{cases}\\
 \begin{cases} 
 p_i\geq 0, ~i\in\{1,2,\hdots,6\}\\
 p_{1}+p_{2}+p_{3}+p_{4}+p_{5}+p_{6}= 1 \\
  (1-\lambda)p_{1}+p_{2}+(\lambda-1)p_{3}- p_{4}+\lambda p_{5}-\lambda p_{6}>0, & A_1~beats~A_2  \\ \lambda p_{1}-\lambda p_{2}+p_{3}+(1-\lambda)p_{4}- p_{5}+(\lambda-1)p_{6}>0, & A_2~beats~A_3  \\  p_{4}+p_{5}+p_{6}-p_1+(\lambda-1)p_2-\lambda p_{3}>0, & d(A_3,A_1)>0 \\
 p_{4}+p_{5}+p_{6}-p_3+\lambda p_2-\lambda p_{1}>0, & d(A_3,A_2)>0\\ (2-\lambda)(p_{4}+p_{5}+p_{6})-(1+\lambda)p_1+(2\lambda-1)p_2-(1+\lambda)p_3>0, & d(A_3,A_1)+d(A_3,A_2)>\lambda\mathfrak{C}_{A_3}.    \end{cases}
 \end{split}
 \end{equation}

Without loss of generality we assume that the initial arrangement is  ($A_1,A_2,A_3$). This means that  $A_1$ has more total points than $A_2$ who in turn has more total points than $A_3$. This explains the two first inequalities of every system. Recall that  when the winner is alternative $A_1$ we have that alternative $A_2$ is unifying for the coalition $Coal_{A_2}$ formed by voters with preferences $(A_2,A_1,A_3)$,$(A_2,A_3,A_1)$, and $(A_3,A_2,A_1)$. If we count the number of points for each alternative prior to the coalition participants' arrangement of all the places except the first which is given to the unifying alternative $A_k$ (i.e., having intermediate preferences of the form $(A_k,?,?)$), the unifying alternative $A_k$ must have the largest number of points. Thus, the two  inequalities $d(A_k, A_i)>0,~i\neq k$ are necessary with $k=2$ and $k=3$ for the first and the second systems, respectively. Note that the inequality called `$A_2$ beats $A_3$' in the first system leads to $d(A_2,A_3)>0$. The next step is to understand when individuals of the coalition unified by an alternative can successfully manipulate the election.   Actually, all the freedom they possess at this point is to give $\lambda$ points from each participant of the coalition (by assigning the second place) to one of the remaining two alternatives without violating the condition that their alternative has more points.  This can be done if and only if the sum of the differences in points between the unifying alternative and each of the remaining two, prior to the coalition members' making a choice of their second most preferable alternative, is greater than  $\lambda$  times the number of people in the coalition. This is equivalent to $d(A_2,A_1)+d(A_2,A_3)>\lambda\mathfrak{C}_{A_2} $ and $d(A_3,A_1)+d(A_3,A_2)>\lambda\mathfrak{C}_{A_3}$ in the first and the second systems, respectively. Recalling that we consider a manipulation by  coalitions of maximal sizes (i.e., $\mathfrak{C}_{A_2}=p_{3}+p_{4}+p_{6}$ and $\mathfrak{C}_{A_3}=p_{4}+p_{5}+p_{6}$), this leads to the last inequality of every system. That analysis gives first insights on the validity of our approach since the systems described above are exactly the same as the ones already obtained by  \cite{MoyouwouT2017}  who  determined necessary and sufficient conditions for a given profile to be coalitionally manipulable under a general scoring rule with three alternatives. 

We consider now the general case of $m$ alternatives. Recall that the unifying alternative can be any alternative except the winner so that $m-1$ coalitions can be formed.  Here again we consider the coalitions of maximal sizes, i.e., each coalition consists of  all voters with the unifying alternative ranked higher than the winner in their preferences. The set of such preferences has cardinality $\frac{m!}{2}$.

\begin{thm}\label{mainThm}
Let $m\geq 3$ and suppose that the initial arrangement is  $(A_1,A_2,\hdots,A_m)$ under the positional voting rule having weight vector $w=(w_1,w_2,\hdots,w_{m})$.   A profile $p=(p_1,\hdots,p_{m!})$ is manipulable by the coalition with unifying alternative $A_k$ if and only if $p$ satisfies the following system of inequalities: 
\begin{equation}
\label{ManipGeneralM}
\begin{cases}
\mbox{Preliminary Inequalities (PI)}
\\
~~~~~~~~~~~~\begin{cases}

 \sum\limits_{i=1}^{m!}p_i=1 \\
p_i\geq 0\\
\mbox{The initial arrangement is } (A_1,A_2,\hdots,A_m) 
\end{cases}\\
\\
\mbox{Strategic Inequalities (SI)}
\\
~~~~~~~~~~~~\begin{cases}
 
\sum\limits_{i\neq k}d(A_k,A_i)>(w_2+\hdots+w_{m-1}+w_{m})\mathfrak{C}_{A_k}\\ \sum\limits_{i\neq k}d(A_k,A_i)-M_1>(w_3+\hdots+w_{m-1}+w_{m})\mathfrak{C}_{A_k}\\\hdots \\\sum\limits_{i\neq k}d(A_k,A_i)-\sum\limits_{j=1}^{m-4}M_j>(w_{m-2}+w_{m-1}+w_m)\mathfrak{C}_{A_k}\\  
\sum\limits_{i\neq k}d(A_k,A_i)-\sum\limits_{j=1}^{m-3}M_j>(w_{m-1}+w_{m})\mathfrak{C}_{A_k}\\ 
M_{m-1}>w_{m}\mathfrak{C}_{A_k}, \\
\end{cases}
\end{cases}
\end{equation}
where $M_1=\mbox{max}\{d(A_k,A_i)~|~i\neq k\}, M_2=\mbox{max}_2\{d(A_k,A_i)~|~i\neq k\},\hdots, M_{m-1}=\mbox{max}_{m-1}\{d(A_k,A_i)~|~i\neq k\}$. The notation  $\mbox{max}_s\{d(A_k,A_i)~|~i\neq k\}$ stands for the $s^{th}$ largest value in $\{d(A_k,A_i)~|~i\neq k\}$.
\end{thm}

\begin{proof}

First let us show that the SI inequalities in \eqref{ManipGeneralM} are necessary. This will be checked by induction on $m$, the base being $m=3$. Notice that it is clearly reasonable for all the members of the coalition to start filling their profiles by updating their intermediate preferences  to $(A_k,?,\hdots,?)$.  The first inequality in the SI system simply guarantees that after that update the  coalition participants can distribute the remaining votes they have to give without violating the fact that $A_k$ has more points than any other alternatives with no restrictions on the distribution of the votes. The second to last inequalities of the SI system can be seen to be those for the positional rule with $w'=(w_1,w_3,\hdots,w_{m-1},w_{m})$ and $m-1$ alternatives (alternative $A_{s}$ with $M_1=d(A_k,A_{s})$ is eliminated). The statement follows.\\ \\
Now we establish the sufficiency of inequalities \eqref{ManipGeneralM}.  We argue by induction on the number of alternatives $m$, the base being $m=3$. The sufficiency of inequalities \eqref{ManipGeneralM} is established via the step by step procedure of updating the intermediate preferences of coalition participants presented below.
\begin{enumerate}
\item[\textit{Step 1}.] All the members of the coalition start filling their preferences by putting $A_k$ as the most preferrable alternative. At this point the preferences of all participants of $Coal_{A_k}$ are identical and equal to $(A_k,?,\hdots,?)$, where the question marks stand for the places yet to be assigned.
\item[\textit{Step 2}.] Let $A_{s}$ be the alternative with $d(A_k,A_{s})=M_1$. We will distinguish between two cases:
\vspace{0.5cm}
\begin{enumerate}
\item [\textit{case (a)}] $M_1> w_2\mathfrak{C}_{A_k}$.  Then all participants of $Coal_{A_k}$ update their intermediate preferences to $(A_k,A_{s},?,\hdots,?)$ and the list of inequalities in the SI system of \eqref{ManipGeneralM} without the second one from the top is easily seen to be the one for the positional rule with $m-1$ alternatives and weight vector $(w_1,w_3,\hdots,w_{m})$;
\vspace{0.5cm}

\item [\textit{case (b)}] 
$w_{\sigma}\mathfrak{C}_{A_k}< M_1<w_{\sigma-1}\mathfrak{C}_{A_k}$ for some $2<\sigma\leq m$. We would like to point out that the case $M_1<w_{m}\mathfrak{C}_{A_k}$ can never occur as it would imply 	$\sum\limits_{i\neq k}d(A_k,A_i)<(m-1)w_{m}\mathfrak{C}_{A_k}<(w_2+\hdots+w_{m-1}+w_{m})\mathfrak{C}_{A_k}$ and violate the first (SI) inequality.

\vspace{0.3cm}
  Let $\varepsilon>0$ be an arbitrary small positive constant, satisfying
\vspace{0.3cm}
\begin{enumerate}
\item[($\star$)] $\varepsilon$ does not exceed the smallest difference between the left-hand side and right-hand side  of the SI inequalities of \eqref{ManipGeneralM}.
\end{enumerate}
\vspace{0.3cm}

It will serve as the remaining difference in points between $A_k$ and $A_s$ after the update of preferences described below.  Set $t:=\frac{M_1-w_{\sigma}\mathfrak{C}_{A_k}}{w_{\sigma-1}-w_{\sigma}}-\varepsilon$ and let $t$ `participants' update their  intermediate preferences to $(A_k,?,\hdots,?,\underset{\sigma-1}{A_s},?,\hdots,?)$ and the remaining  $\mathfrak{C}_{A_k}-t$ to $(A_k,?,\hdots,?,\underset{\sigma}{A_s},?,\hdots,?)$. Introduce the weight 
\begin{equation}
w_{\sigma-1,\sigma}:=\frac{(\mathfrak{C}_{A_k}-t)w_{\sigma-1}+tw_{\sigma}}{\mathfrak{C}_{A_k}}
\end{equation}
and set 
\begin{equation}
w'=(w_1,\hdots,w_{\sigma-2},w_{\sigma-1,\sigma},w_{\sigma+1},\hdots,w_{m}). 
\end{equation}
After updates of coalition participants' intermediate preferences and elimination of alternative $A_s$, the system of SI inequalities in \eqref{ManipGeneralM} becomes the one for the positional rule with $m-1$ alternatives and weight vector $w'$:
\begin{equation}
\begin{cases}
\sum\limits_{i\neq k,s}d(A_k,A_i)>(w_2+\hdots+w_{\sigma-1,\sigma}+\hdots+w_m)\mathfrak{C}_{A_k}\\ 
\sum\limits_{i\neq k,s}d(A_k,A_i)-M_2>(w_3+\hdots+w_{\sigma-1,\sigma}+\hdots+w_m)\mathfrak{C}_{A_k}\\
\hdots \\
\sum\limits_{i\neq k,s}d(A_k,A_i)-\sum\limits_{j=2}^{\sigma-2}M_j>(w_{\sigma-1,\sigma}+\hdots+w_{m})\mathfrak{C}_{A_k}\\
\sum\limits_{i\neq k,s}d(A_k,A_i)-\sum\limits_{j=2}^{\sigma-1}M_j>(w_{\sigma+1}+\hdots+w_m)\mathfrak{C}_{A_k}\\
\hdots \\
\sum\limits_{i\neq k,s}d(A_k,A_i)-\sum\limits_{j=2}^{m-4}M_j>(w_{m-2}+w_{m-1}+w_m)\mathfrak{C}_{A_k}\\  
\sum\limits_{i\neq k,s}d(A_k,A_i)-\sum\limits_{j=2}^{m-3}M_j>(w_{m-1}+w_m)\mathfrak{C}_{A_k}\\
M_{m-1}>w_m\mathfrak{C}_{A_k}, 
\end{cases}
\end{equation}

which, using that $M_1=d(A_k,A_s)=(w_{\sigma-1}+w_{\sigma}-w_{\sigma-1,\sigma})\mathfrak{C}_{A_k}$, can be rewritten as

\begin{equation}\label{intermStep}
\begin{cases}
\sum\limits_{i\neq k}d(A_k,A_i)>(w_2+\hdots+w_{\sigma-1}+w_{\sigma}+\hdots+w_m)\mathfrak{C}_{A_k}+\varepsilon \\ 
\sum\limits_{i\neq k}d(A_k,A_i)-M_2>(w_3+\hdots+w_{\sigma-1}+w_{\sigma}+\hdots+w_m)\mathfrak{C}_{A_k}+\varepsilon \\
\hdots \\
\sum\limits_{i\neq k}d(A_k,A_i)-\sum\limits_{j=2}^{\sigma-2}M_j>(w_{\sigma-1}+w_{\sigma}+\hdots+w_m)\mathfrak{C}_{A_k}+\varepsilon \\
\sum\limits_{i\neq k}d(A_k,A_i)-\sum\limits_{j=2}^{\sigma-1}M_j>(w_{\sigma+1}+\hdots+w_m)\mathfrak{C}_{A_k} \\
\hdots \\
\sum\limits_{i\neq k}d(A_k,A_i)-\sum\limits_{j=1}^{m-4}M_j>(w_{m-2}+w_{m-1}+w_m)\mathfrak{C}_{A_k} \\  
\sum\limits_{i\neq k}d(A_k,A_i)-\sum\limits_{j=1}^{m-3}M_j>(w_{m-1}+w_m)\mathfrak{C}_{A_k} \\
M_{m-1}>w_m\mathfrak{C}_{A_k}
\end{cases}
\end{equation}

Recalling that $M_1$ is greater than any other $M_i$ by definition and using the $(\star)$ assumption on $\varepsilon$, we see that each of the inequalities in \eqref{intermStep} follows from the corresponding SI  inequality in \eqref{ManipGeneralM}.

\end{enumerate}
\end{enumerate}
\end{proof}

Let us now illustrate the results of this theorem using an example.

\begin{ex}
Consider the Borda rule with weight vector $w=(3,2,1,0)$ for  $m=4$ alternatives $A_1$, $A_2$, $A_3$, and $A_4$ and the voting profile given in Table \ref{ExampleTable} below. In this example, the initial arrangement  is  $(A_1,A_2,A_3,A_4)$ as it is shown in the last row.  We would like to determine if the chosen profile is vulnerable to manipulability by the largest possible coalition with unifying alternative $A_2$, i.e., $Coal_{A_2}$ of cardinality $\mathfrak{C}_{A_2}=p_7+p_8+p_{15}=5/9$. 
\begin{table}[H]
\caption{Voting profile and corresponding point distribution}\label{ExampleTable}\begin{center}
\vspace{-0.3cm}\begin{tabular}{ |c|c|c|c|c|c|} 
\hline
Preference & Fraction & $A_1$ &  $A_2$ &  $A_3$ & $A_4$ \\ 
\hline
$(A_2,A_1,A_3,A_4)$ & $p_7=2/9$ & $4/9$ &$6/9$ & $2/9$ & $0$\\ 
\hline
$(A_2,A_1,A_4,A_3)$ & $p_8=2/9$ & $4/9$ &$6/9$ & $0$ &$2/9$  \\ 
\hline
$(A_3,A_1,A_4,A_2)$ & $p_{14}=1/9$ & $2/9$ &$0$ & $3/9$ &$1/9$ \\ 
\hline
$(A_3,A_2,A_1,A_4)$ & $p_{15}=1/9$ & $1/9$ &$2/9$ & $3/9$ &$0$ \\ 
\hline
$(A_3,A_4,A_1,A_2)$ & $p_{17}=1/9$ & $1/9$ &$0$ & $3/9$ &$2/9$ \\ 
\hline
$(A_4,A_1,A_3,A_2)$ & $p_{20}=2/9$ & $4/9$ &$0$ & $2/9$ & $6/9$\\ 
\hline
Other & $p_i=0$ & $0$ &$0$ & $0$ &$0$ \\ 
\hline
$\sum$ & $1$ & $16/9$ &$14/9$ & $13/9$ &$11/9$ \\ 
\hline
\end{tabular}
\end{center}
\end{table}

In order to answer this question we will adhere to the step by step procedure given in the proof of Theorem \ref{mainThm} above. The intermediate preferences of the participants in the coalition after the first update are presented in Table  \ref{ExampleInterm1Table}. 

\begin{table}[H]
\caption{Step $1$ intermediate preferences and corresponding point distribution}
\label{ExampleInterm1Table}
\begin{center}
\vspace{-0.3cm}\begin{tabular}{ |c|c|c|c|c|c|} 
\hline
Preference & Fraction & $A_1$ &  $A_2$ &  $A_3$ & $A_4$ \\ 
\hline
$(A_2,?,?,?)$ & $\mathfrak{C}_{A_2}=5/9$ & $0$ &$15/9$ & $0$ & $0$\\  
\hline
$(A_3,A_1,A_4,A_2)$ & $p_{14}=1/9$ & $2/9$ &$0$ & $3/9$ &$1/9$ \\ 
\hline
$(A_3,A_4,A_1,A_2)$ & $p_{17}=1/9$ & $1/9$ &$0$ & $3/9$ &$2/9$ \\ 
\hline
$(A_4,A_1,A_3,A_2)$ & $p_{20}=2/9$ & $4/9$ &$0$ & $2/9$ & $6/9$\\ 
\hline
Other & $p_i=0$ & $0$ &$0$ & $0$ &$0$ \\ 
\hline
$\sum$ & $1$ & $7/9$ &$15/9$ & $8/9$ &$1$ \\ 
\hline
\end{tabular}
\end{center}
\end{table}

This allows to evaluate $d(A_2,A_i)$ for every $i\neq 2$:  

$d(A_2,A_1)=15/9-7/9=8/9$ 

$d(A_2,A_3)=15/9-8/9=7/9 $

$d(A_2,A_4)=15/9-1=6/9$

From here we deduce that $M_1=\mbox{max}\{d(A_2,A_i)~|~i\neq k\} = d(A_2,A_1)=8/9$ and since $w_3\mathfrak{C}_{A_2}=5/9<8/9<10/9=w_2\mathfrak{C}_{A_2},$ we are in case $(b)$ of Step $(2)$ above. The next step is to compute 

$$t=\frac{8/9-1\cdot5/9}{2-1}-\varepsilon=3/9-\varepsilon$$
and update the intermediate profiles accordingly (see Table \ref{ExampleInterm2Table}). At this point we have reduced the problem to the case of three alternatives $(A_2, A_3, \mbox{ and } A_4)$ and positional rule determined by the weight vector $(3,w_{2,3},0)$, where
$$w_{2,3}=\frac{(2/9+\varepsilon)\cdot1+(3/9-\varepsilon)\cdot 2}{5/9}=8/5-9\varepsilon/5=1.6-9\varepsilon/5.$$

\begin{table}[H]
\caption{Step $2$ intermediate preferences and corresponding point distribution}
\label{ExampleInterm2Table}
\begin{center}
\vspace{-0.3cm}\begin{tabular}{ |c|c|c|c|c|c|} 
\hline
Preference & Fraction & $A_1$ &  $A_2$ &  $A_3$ & $A_4$ \\ 
\hline
$(A_2,A_1,?,?)$ & $t=3/9-\varepsilon$ & $6/9-2\varepsilon$ &$1-3\varepsilon$ & $0$ & $0$\\  
\hline
$(A_2,?,A_1,?)$ & $2/9+\varepsilon$ & $2/9+\varepsilon$ &$6/9+3\varepsilon$ & $0$ & $0$\\  
\hline
$(A_3,A_1,A_4,A_2)$ & $p_{14}=1/9$ & $2/9$ &$0$ & $3/9$ &$1/9$ \\ 
\hline
$(A_3,A_4,A_1,A_2)$ & $p_{17}=1/9$ & $1/9$ &$0$ & $3/9$ &$2/9$ \\ 
\hline
$(A_4,A_1,A_3,A_2)$ & $p_{20}=2/9$ & $4/9$ &$0$ & $2/9$ & $6/9$\\ 
\hline
Other & $p_i=0$ & $0$ &$0$ & $0$ &$0$ \\ 
\hline
$\sum$ & $1$ & $15/9-\varepsilon$ &$15/9$ & $8/9$ &$1$ \\ 
\hline
\end{tabular}
\end{center}
\end{table}

At this stage the question of manipulability is resolved according to whether or not the inequality $$d(A_2,A_3)+d(A_2,A_4)>w_{2,3}\mathfrak{C}_{A_2}$$ is satisfied. Plugging in the corresponding values we obtain $$7/9+6/9=13/9>8/9-\varepsilon=5/9(1.6-9\varepsilon/5),$$ securing the affirmative answer. 

One possible arrangement of the complete profiles is presented in Table \ref{ExampleInterm3Table} and this concludes our example.

\begin{table}[H]
\caption{Complete strategic profile and resulting point distribution}\label{ExampleInterm3Table}
\begin{center}
\vspace{-0.7cm}\begin{tabular}{ |c|c|c|c|c|c|} 
\hline
Preference & Fraction & $A_1$ &  $A_2$ &  $A_3$ & $A_4$ \\ 
\hline
$(A_2,A_1,A_3,A_4)$ & $(3/9-\varepsilon)\cdot3/5$ & $18/45-6\varepsilon/5$ &$27/45-9\varepsilon/5$ & $9/45-3\varepsilon/5$ & $0$\\  
\hline
$(A_2,A_3,A_1,A_4)$ & $(2/9+\varepsilon)\cdot3/5$ & $6/45+3\varepsilon/5$ &$18/45+9\varepsilon/5$ & $12/45+6\varepsilon/5$ & $0$\\  
\hline
$(A_2,A_1,A_4,A_3)$ & $(3/9-\varepsilon)\cdot2/5$ & $12/45-4\varepsilon/5$ &$18/45-6\varepsilon/5$ & $0$ & $6/45-2\varepsilon/5$\\  
\hline
$(A_2,A_4,A_1,A_3)$ & $(2/9+\varepsilon)\cdot2/5$ & $4/45+2\varepsilon/5$ &$12/45+6\varepsilon/5$ & $0$ & $8/45+4\varepsilon/5$\\  
\hline
$(A_3,A_1,A_4,A_2)$ & $p_{14}=1/9$ & $2/9$ &$0$ & $3/9$ &$1/9$ \\ 
\hline
$(A_3,A_4,A_1,A_2)$ & $p_{17}=1/9$ & $1/9$ &$0$ & $3/9$ &$2/9$ \\ 
\hline
$(A_4,A_1,A_3,A_2)$ & $p_{20}=2/9$ & $4/9$ &$0$ & $2/9$ & $6/9$\\ 
\hline
Else & $p_i=0$ & $0$ &$0$ & $0$ &$0$ \\ 
\hline
$\sum$ & $1$ & $75/45-\varepsilon$ &$75/45$ & $61/45+3\varepsilon/5$ &$59/45+2\varepsilon/5$ \\ 
\hline
\end{tabular}

\end{center}
\end{table}

\end{ex}

Since $w_2 = w_3 = \hdots=w_m=0$ under the Plurality rule, the  $SI$ inequalities reduce to $d(A_k,A_i)>0 \mbox{ for } i\in\{1,\hdots,m~|~i\neq k\}$. Then the following corollary can be deduced from Theorem \ref{ManipGeneralM}. 

\begin{cor}
\label{PlurCor}
Consider the Plurality voting rule with $m\geq 3$ alternatives. A profile $p=(p_1,\hdots,p_{m!})$ is manipulable by a coalition with unifying alternative $A_k$ if and only if $p$ satisfies the following system of inequalities:
\begin{equation}
\label{ManipGeneralPlur}
\begin{cases} \sum\limits_{i=1}^{m!}p_i=1 \\
p_i\geq 0\\
\mbox{The initial arrangement is } (A_1,A_2,\hdots,A_m) \\ 
d(A_k,A_i)>0 \mbox{ for } i\in\{1,\hdots,m~|~i\neq k\}  ~~~~~~~~~(SI).\\ 
\end{cases}
\end{equation} 
\end{cor}

We would like to bring the reader's attention to the fact that the $SI$ inequalities in \eqref{ManipGeneralPlur} are the same as inequalities $(2)$ in Theorem $1$ of  \cite{LM87}, where the characterization of coalitional manipulable profiles under Plurality rule was first derived.  We also bring the reader's attention to the fact that  \cite{Lepelley3}  and  \cite{Favardin} deal with the coalitional manipulation of the Antiplurality rule ($\lambda=1$)  and  the  Borda rule ($\lambda=2$), respectively, in the case of three-alternative elections.  The approach that we use in this paper to derive our list of inequalities is different from those used by  \cite{Lepelley3}  and  \cite{Favardin}. This explains  why the systems given in these papers (Lemma 2 in \cite{Lepelley3} and Lemma 4 in \cite{Favardin})  are not a direct consequence of the systems given in Theorem \ref{mainThm}.  However, our results  in  Section \ref{sectionproba}  for the case of three alternatives  allow us to show that the two characterizations are equivalent.  

As already noticed, the approach that we consider in the main result of the paper can be adapted when only manipulation by small coalitions is considered. Indeed, under some profiles a small coalition may reverse the relative ranking of two alternatives as it is shown in the following example. 

\begin{ex} Let $m=3$ and the scoring rule be Borda rule having $w=(2,1,0)$. As in Example \ref{BordaManipEx},  we set  $p_{1}=5/9$,  $p_{3}=4/9$, and  $p_{j}=0$  for $j\neq 1,3$. The initial arrangement is  $(A_1,A_2,A_3)$ since alternative $A_1$ wins and gets $14/9$ points, $A_2$ gets $13/9$ points,  and finally $A_3$ gets $0$ points.  But, if a fraction $\mathfrak{p}=2/9$ of voters chose $(A_2,A_3,A_1)$ instead of $(A_2,A_1,A_3)$, and the others kept their preferences unchanged, then the final result would be in favor of alternative $A_2$ since $A_1$ gets $12/9$ points, $A_2$ gets $13/9$ points, and $A_3$ gets $2/9$ points.  In other words, a manipulation in favor of $A_2$ is also possible with a coalition of voters smaller than the maximal size (see Example \ref{BordaManipEx} for a manipulation by a coalition of a maximal size).     \end{ex}

\begin{cor}
Let $m\geq 3$ and suppose that the initial arrangement is  $(A_1,A_2,\hdots,A_m)$ under the positional voting rule having weight vector $w=(w_1,w_2,\hdots,w_{m})$. A profile $p=(p_1,\hdots,p_{m!})$	is manipulable by some coalition not exceeding a proportion $0<\mathfrak{p}\leq 1$ of the electorate with unifying alternative $A_k$ if and only if $p$ satisfies the following system of inequalities: 
	
	\begin{equation}
	\label{ManipGeneralMSmallerCoal}
	\begin{cases}
	\mbox{Preliminary Inequalities (PI)}
	\\
	~~~~~~~~~~~~\begin{cases}
	
	\sum\limits_{i=1}^{m!}p_i=1 \\
	p_i\geq 0\\
	0\leq \widetilde{p}_j\leq p_j \mbox{ for } j, \mbox{ s.t. } A_1\prec_j A_k\\
	\sum\limits_{A_1\prec_j A_k}\widetilde{p}_j=\mathfrak{p} \\
	\mbox{The initial arrangement is } (A_1,A_2,\hdots,A_m) 
	\end{cases}\\
	\\
	\mbox{Strategic Inequalities (SI)}
	\\
	~~~~~~~~~~~~\begin{cases}
	
	\sum\limits_{i\neq k}d(A_k,A_i)>(w_2+\hdots+w_{m-1}+w_{m})\mathfrak{p}\\ \sum\limits_{i\neq k}d(A_k,A_i)-M_1>(w_3+\hdots+w_{m-1}+w_{m})\mathfrak{p}\\\hdots \\\sum\limits_{i\neq k}d(A_k,A_i)-\sum\limits_{j=1}^{m-4}M_j>(w_{m-2}+w_{m-1}+w_m)\mathfrak{p}\\  
	\sum\limits_{i\neq k}d(A_k,A_i)-\sum\limits_{j=1}^{m-3}M_j>(w_{m-1}+w_{m})\mathfrak{p}\\ 
	M_{m-1}>w_{m}\mathfrak{p}, \\
	\end{cases}
	\end{cases}
	\end{equation}
	where by $A_1\prec_j A_k$ we understand that $A_k$ has a higher ranking than $A_1$ in preference of type  $j$ and $\widetilde{p}_j$ denotes the share of `participants'  with preference of this type. The added $PI $ inequalities correspond to the fact that each $\widetilde{p}_j$ can not exceed $p_j$. As it is sufficient to consider coalitions of the largest allowed proportion, we add the equality $\sum\limits_{A_1\prec_j A_k}\widetilde{p}_j=\mathfrak{p}$. 
\end{cor}

The next remark concerns the case where the number of voters is small.

\begin{rmk}
	\label{finitermk}
Theorem \ref{mainThm} does not hold true in case of finite number of voters. Consider the case of $m=3$ alternatives $A_1$, $A_2$  and  $A_3$ and the scoring rule with weight vector $w=(1,0.9,0)$. Let there be $n=21$ voters and the profile $p=(6,7,8,0,0,0)$. Then initially $A_1$ gets $20.2$ points, $A_2$ gets $13.4$ points and $A_3$ gets $6.3$ points satisfying the arrangement assumptions. Next we show that the remaining two inequalities in the system, responsible for the possibility of manipulation by $Coal_{A_2}$,  hold true for the profile $p$ as well. Indeed, $d(A_2,A_1)=0.4>0$ and $d(A_2,A_1)+d(A_2,A_3)=7.5>7.2=\lambda \mathfrak{C}_{A_2}$. However, if a single member of the coalition updated his/her profile to $(A_2,A_1,A_3)$ then $\lambda=0.9>0.4=d(A_2,A_1)$ would lead to alternative $A_1$ getting ahead of $A_2$. On the other hand it is not possible for all $8$ people in $Coal_{A_2}$ to update their profiles to $(A_2,A_3,A_1)$ either, as $d(A_2,A_3)=7.1<7.2=\lambda \mathfrak{C}_{A_2}$. Hence, the profile is not manipulable by $Coal_{A_2}$.
\end{rmk}

We finish this part with a remark on the manipulability of the Antiplurality rule.

\begin{rmk}\label{NonManipRMK}
Suppose that the initial arrangement  is  $(A_1,A_2,\hdots,A_m)$. The group of voters having the last alternative prior in their preferences $(Coal_{A_m})$  can never manipulate an election under Antiplurality  rule whichever the number of alternatives $m$ is. Otherwise, as after such a manipulation the last alternative $A_m$ would have the same number of points (his number of points is not influenced by the manipulation) all the other alternatives would have to undergo a simultaneous loss of points. That is clearly impossible, because the number of rearranged points equals to the number of points initially arranged between the alternatives  $A_1,\hdots,A_{m-1}$. 
\end{rmk}

\section{Some numerical results}\label{sectionproba}

The aim of this section is to compute the share of manipulable profiles for three well-studied scoring rules under the assumption of Impartial and Anonymous Culture (IAC) in  the limit case as $n\to \infty$.  We consider this hypothesis because it has been widely used in many papers that we cite in the introduction.  When $n\to \infty$,  this assumption indicates  that all profile vectors of the form $p=(p_1,\hdots,p_{m!})$,  with each $p_i\geq 0$ and $\sum\limits_{i=1}^{m!}p_i= 1$, are equally likely to occur. This probability distribution was first introduced by  \cite{GF}. For more details on the IAC condition and others, we refer the reader to \cite{DM20,GehrleinLepelley1,GehrleinLepelley2}.  

It is well-known that the limiting probability of every voting event (here the manipulation of voting outcomes) under IAC can be simply reduced to volume computations.   Since the normalized profile vectors are of the form $p=(p_1,\hdots,p_{m!})$ with each $p_i\geq 0$ and $\sum\limits_{i=1}^{m!}p_i= 1$, the polytope that describes all possible profile vectors defines the standard simplex $\triangle \subset \mathbb{R}^{m!}$. The volume of $\triangle \subset \mathbb{R}^d$ is equal to $\frac{1}{(d-1)!}$, which, since in our case $d=m!$, becomes $\frac{1}{m!-1}$.  What remains to compute is the volumes of the polytopes $\mathcal{P}_{A_k}$ (with $k=1,\dots,m$)   cutting out  the regions of profiles manipulable by $Coal_{A_k}$ inside the standard simplex $\triangle$. For instance, in the case of three alternatives, the volumes of the polytopes $\mathcal{P}_{A_2}$   and  $\mathcal{P}_{A_3}$ cutting out  the regions of profiles manipulable by $Coal_{A_2}$ and $Coal_{A_3}$ (see systems in (\ref{ManipEq}))  inside the standard simplex  $\triangle$ are required.  Then the total volume of the region of manipulable profiles is given by $Vol(\mathcal{P}_{A_2})+Vol(\mathcal{P}_{A_3})- Vol(\mathcal{P}_{A_2}\cap\mathcal{P}_{A_3})$.  Afterwards the result is divided by the volume of the standard simplex  and multiplied by $6$ (all possible final arrangements of alternatives) to produce the final answer.
The approach described in \cite{CGZ}\footnote{See also \cite{MoyouwouT2017}.}  tells us an effective way to tackle this problem for three-alternative elections. When we have more than three alternatives, however, those methods  become too complicated in practice.  Indeed, due to dimension reasons, the precise probabilities  seem to be difficult to obtain for $m=4$ and it is impossible to compute for $m\geq 5$ at this stage using the existing algorithms.

In case the precise volume cannot be found, it is natural to obtain a sufficiently good approximation of it. One of the most commonly used procedures for this purpose is the Monte Carlo volume estimation.  The share of manipulable outcomes for Plurality, Antiplurality and Borda rule in case of $m=3,4,5$ alternatives are presented in Table \ref{3CandTable} (exact results),  and Tables \ref{4CandTable} and  \ref{5CandTable} (simulated results).

\begin{table}[H]
\caption{Share of manipulable outcomes with $m=3$ alternatives}
\label{3CandTable}
\begin{center}
\vspace{-0.3cm}\begin{tabular}{ |c|c|c|c| } 
\hline
Rule & Manipulable &  by $Coal_{A_2}$ & by $Coal_{A_3}$ \\ 
\hline
Plurality &  $29.17 \%$  & $24.65 \%$ & $15.63 \%$\\ 
\hline
Antiplurality  & $51.85\%$  & $51.85\%$ &      $0\%$ \\ 
\hline
Borda  & $50.25\%$ & $47.71\%$  & $9.71\%$  \\ 
\hline
\end{tabular}
\end{center}
\end{table}

\begin{table}[H]
\caption{Share of manipulable outcomes with $m=4$ alternatives}
\label{4CandTable}
\begin{center}
\vspace{-0.3cm}\begin{tabular}{ |c|c|c|c|c| } 
\hline
Rule & Manipulable &  by $Coal_{A_2}$ & by $Coal_{A_3}$ & by $Coal_{A_4}$\\ 
\hline
Plurality & $87.38\%$ & $83.65\%$ &$73.87\%$ & $63.53\%$ \\ 
\hline
Antiplurality  & $87.13\%$ & $86.47\%$ &$22.83\%$ & $0 \%$ \\ 
\hline
Borda  & $95.65\%$ & $95.03\%$  & $79.16\%$ & $43.38\%$ \\ 
\hline
\end{tabular}
\end{center}
\end{table}

\begin{table}[H]
\caption{Share of manipulable outcomes with $m=5$ alternatives}\label{5CandTable}
\begin{center}
\vspace{-0.3cm}\begin{tabular}{ |c|c|c|c|c|c| } 
\hline
Rule & Manipulable & by $Coal_{A_2}$ & by $Coal_{A_3}$ & by $Coal_{A_4}$ & by $Coal_{A_5}$\\ 
\hline
Plurality & $99.51\%$ & $99.37\%$ &$99.05\%$ & $98.57\%$ & $98.04\%$ \\ 
\hline
Antiplurality & $97.15\%$ & $96.79\%$ &$54.78\%$ & $ 6.52\% $ & $0\%$ \\ 
\hline
Borda  & $99.76\%$ & $99.23\%$  & $98.95\%$ & $98.18\%$& $97.83\%$ \\ 
\hline
\end{tabular}

\end{center}
\end{table}

Our numerical results allow us to conclude that (i) moving from three to four and five alternatives very significantly increases the vulnerability to coalitional manipulation under all the scoring rules that we consider and (ii) the hierarchy of the scoring rules with respect to the susceptibility to coalitional manipulability changes when the number of alternatives increases: The Antiplurality rule seems to be the most performing when the number of alternatives is four and five, whereas it is the most prone to manipulation in the case of three alternatives. 

Recall that  \cite{MoyouwouT2017}  determined necessary and sufficient conditions for a given profile to be coalitionally manipulable under a general scoring rule with three alternatives. The share of manipulable profiles has also been given in this paper under the IAC condition.  Our results for Table  \ref{3CandTable}  coincide with their results which were already obtained  in the literature as in  \cite{LM87,Lepelley},  and \cite{Wilson1}. The case $m=4$ for the Plurality rule was studied in \cite{Ouafdi1}. The precise value of $87.28\%$ for general manipulability was obtained (this corresponds to $87.38\%$ appearing in the top left corner of Table \ref{4CandTable}). In other words, our approximation differs by $0.1\%$. Note also that the zero value in the last column of the Antiplurality rule is consistent with Remark \ref{NonManipRMK}.

Finally, the Python code for carrying out the computation of the share of manipulable outcomes in case of Borda rule with $m=4$ alternatives is provided in Appendix $A$.   A detailed outline on the execution complimentary to the comments inside the code is given. The recipe for finding the corresponding results for other scoring voting rules with  than $m\geq 4$ alternatives is provided.  

\section{Conclusion}\label{concl}

At least three main extensions can emerge from the present study. First, it would be interesting to investigate other voting rules not yet considered, especially in the class of the so-called \emph{iterative scoring rules}, also known as \emph{multi-stage} or \emph{sequential elimination scoring rules}. This class of multi-round procedures are based on the same scoring principle previously defined but proceed by eliminating one or more alternatives at each round, until there is only one alternative left who is considered as  the winner. Second, taking into account the
possibility of reactions of voters outside the coalition that manipulates the vote may also be a good extension of this study. This appears to have a considerable impact on the manipulability shares and on the induced hierarchy of voting rules when analyzing strategic voting (see for instance \cite{FavardinLepelley}). Third, our probabilistic investigations may also be extended to the framework with small sizes of manipulating coalitions using  Corollary  \ref{ManipGeneralMSmallerCoal}.

\subsection*{Appendix A: Python code used for computations in case of Borda count with $m=4$ alternatives} 

In this section we present the code (written in Python) used for carrying out the algorithm for finding the share of manipulable outcomes in case of Borda rule with $m=4$ alternatives.  With a few modifications to the code provided one can  perform analogous computations for scoring voting rules with different weights. 

\begin{enumerate}
	\item[\textit{Step 1}.] Create the array of possible preferences with the alternatives corresponding to numbers $0$ to $3$ arranged in increasing order. This array has $24$ elements. Each element of the preferences' array is an array of its own with $4$ elements (see lines $14-25$ in the code). 
	\begin{align*} 
	&\mbox{preferences}[0]=\{[0],[1],[2],[3]\},\\
	&\mbox{preferences}[1]=\{[0],[1],[3],[2]\},\\
	&\hdots\\
	&\mbox{preferences}[23]=\{[3],[2],[1],[0]\}.
	\end{align*}
	After this the point distribution array is created, i.e., the element $(i,j)$ (here $0\leq i \leq 23$ and  $0\leq j \leq 3$) is the number of points alternative $j$ gets from preference $i$ (lines $27-37$.)
	\item[\textit{Step 2}.] Now we are in position to store the inequalities responsible for the initial arrangement of alternatives. As usual we assume that this arrangement is $(0,1,2,3)$ and the $3$ PI inequalities responsible for that are obtained (lines $3-8$ and $38-41$).
	\item[\textit{Step 3}.] Next generate a sample of $n$  (`trials' in the code) random points in the simplex $\triangle$. For this choose $24$ random numbers on the closed interval $[0;1]$, arrange them in a nondecreasing order and take the subsequent differences. The points created this way have a uniform distribution in  the simplex $\triangle$ (see Chapter $\rom{5}.2$ in \cite{Devroye}).  The corresponding lines of the code are $43-54$.
	\item[\textit{Step 4}.] The coalition arrays are established. These are $3$ arrays of $24$ elements, each element equal to $1$ or $0$, depending on whether preference $i$ belongs to the coalition or not (lines $55-66$).
	\item[\textit{Step 5}.] The difference arrays are generated (lines $67-105$). 
	\item[\textit{Step 6}.] The remaining part of the code stores the SI inequalities \eqref{ManipGeneralM}  responsible for the profile being manipulable by each of the coalitions separately (recall that  inequalities \eqref{ManipGeneralM} are written in terms of the elements of difference arrays established on the previous step) and then checks for all sample points to determine how many of them satisfy the initial arrangement conditions and are manipulable by at least one coalition.
\end{enumerate}

\begin{small}
\begin{lstlisting}[language=Python]
import random

InitArrangement = 3 * [0]

for i in range(0, 3):
	InitArrangement[i] = []
	for j in range(0, 24):
		InitArrangement[i].append(0)
# generating array of profiles
preferences = []
for i in range(1, 25):
	preferences.append(i)
counter = 0
for i in range(0, 4):
	for j in range(0, 4):
		for k in range(0, 4):
			for l in range(0, 4):
				if (i != j and i != k and i != l
				and j != k and j != l and k != l):
					preferences[counter] = []
					preferences[counter].append(i)
					preferences[counter].append(j)
					preferences[counter].append(k)
					preferences[counter].append(l)
					counter += 1
# generating array of points distribution according to profiles
PointDistArr = []
for i in range(0, 4):
	PointDistArr.append(i)
for i in range(0, 4):
	PointDistArr[i] = []
for j in range(0, 24):
	PointDistArr[i].append(0)
for i in range(0, 24):
	for j in range(0, 4):
		s = preferences[i][j]
		PointDistArr[s][i] = 3 - j
for j in range(0, 24):  # Initial arrangement is A>B>C>D
	InitArrangement[0][j] = PointDistArr[0][j] - PointDistArr[1][j]
	InitArrangement[1][j] = PointDistArr[1][j] - PointDistArr[2][j]
	InitArrangement[2][j] = PointDistArr[2][j] - PointDistArr[3][j]

trials = 8000000  # number of points in the sample
TrialPts = trials * [0]
S = 24 * [0]
# creating random sample of points in the simplex
for i in range(0, trials):
	TrialPts[i] = []
	for j in range(0, 23):
		S[j] = random.uniform(0, 1)
	S[23]=1
	TrialPts[i] = sorted(S)
	for k in range(1, 24):
		TrialPts[i][24 - k] -= TrialPts[i][23 - k]
# initiating coalition members arrays (1 if the preference
# belongs to the coalition and 0 otherwise)
CoalB = 24 * [0]
CoalC = 24 * [0]
CoalD = 24 * [0]
for j in range(0, 24):
	if PointDistArr[0][j] < PointDistArr[1][j]:
		CoalB[j] = 1
	if PointDistArr[0][j] < PointDistArr[2][j]:
		CoalC[j] = 1
	if PointDistArr[0][j] < PointDistArr[3][j]:
		CoalD[j] = 1
dBA = 24 * [0]
dBC = 24 * [0]
dBD = 24 * [0]
dCA = 24 * [0]
dCB = 24 * [0]
dCD = 24 * [0]
dDA = 24 * [0]
dDB = 24 * [0]
dDC = 24 * [0]
# computing d(B,A), d(B,C) and d(B,D)
for j in range(0, 24):
	if CoalB[j] == 1:
		dBA[j] = 3
		dBC[j] = 3
		dBD[j] = 3
	else:
		dBA[j] = PointDistArr[1][j] - PointDistArr[0][j]
		dBC[j] = PointDistArr[1][j] - PointDistArr[2][j]
		dBD[j] = PointDistArr[1][j] - PointDistArr[3][j]
# computing d(C,A), d(C,B) and d(C,D)
for j in range(0, 24):
	if CoalC[j] == 1:
		dCA[j] = 3
		dCB[j] = 3
		dCD[j] = 3
	else:
		dCA[j] = PointDistArr[2][j] - PointDistArr[0][j]
		dCB[j] = PointDistArr[2][j] - PointDistArr[1][j]
		dCD[j] = PointDistArr[2][j] - PointDistArr[3][j]
# computing d(D,A), d(D,B) and d(D,C)
for j in range(0, 24):
	if CoalD[j] == 1:
		dDA[j] = 3
		dDB[j] = 3
		dDC[j] = 3
	else:
		dDA[j] = PointDistArr[3][j] - PointDistArr[0][j]
		dDB[j] = PointDistArr[3][j] - PointDistArr[1][j]
		dDC[j] = PointDistArr[3][j] - PointDistArr[2][j]

counter = 0
B = 3*[0]
C = 3*[0]
D = 3*[0]
Sb = trials*[0]
Sc = trials*[0]
Sd = trials*[0]
BmanipCheck = False
CmanipCheck = False
DmanipCheck = False
CoalBsize = 0
CoalCsize = 0
CoalDsize = 0
for i in range(0, trials):
	check = False
	CoalBsize = 0
	CoalCsize = 0
	CoalDsize = 0
	Sb[i] = []
	Sc[i] = []
	Sd[i] = []
	for j in range(0, 3):
		s = 0
		check = True
		for k in range(0, 24):
			s += TrialPts[i][k] * InitArrangement[j][k]
		if s < 0:
			check = False
			break
	if check is True:
		BmanipCheck = False
		CmanipCheck = False
		DmanipCheck = False
		for j in range(0, 3):
			B[j] = 0
			C[j] = 0
			D[j] = 0
		for k in range(0, 24):
			CoalBsize += TrialPts[i][k] * CoalB[k]
			CoalCsize += TrialPts[i][k] * CoalC[k]
			CoalDsize += TrialPts[i][k] * CoalD[k]
			B[0] += TrialPts[i][k] * dBA[k]
			B[1] += TrialPts[i][k] * dBC[k]
			B[2] += TrialPts[i][k] * dBD[k]
			C[0] += TrialPts[i][k] * dCA[k]
			C[1] += TrialPts[i][k] * dCB[k]
			C[2] += TrialPts[i][k] * dCD[k]
			D[0] += TrialPts[i][k] * dDA[k]
			D[1] += TrialPts[i][k] * dDB[k]
			D[2] += TrialPts[i][k] * dDC[k]
		Sb[i] = sorted(B)
		Sc[i] = sorted(C)
		Sd[i] = sorted(D)
		
		if (Sb[i][0] > 0) and (Sb[i][0] + Sb[i][1] > CoalBsize) and (Sb[i][0] + Sb[i][1] + Sb[i][2] > 3 * CoalBsize):
			BmanipCheck = True
		if (Sc[i][0] > 0) and (Sc[i][0] + Sc[i][1] > CoalCsize) and (Sc[i][0] + Sc[i][1] + Sc[i][2] > 3 * CoalCsize):
			CmanipCheck = True
		if (Sd[i][0] > 0) and (Sd[i][0] + Sd[i][1] > CoalDsize) and (Sd[i][0] + Sd[i][1] + Sd[i][2] > 3 * CoalDsize):
			DmanipCheck = True
	
	if check is True and (BmanipCheck is True or CmanipCheck is True or DmanipCheck is True):
	counter += 1
ans = 0.0
ans = float(24 * counter / 10000)
print(str(ans) + '%')
\end{lstlisting}
\end{small}

\bibliographystyle{amsplain}

\begin{thebibliography}{99}

\bibitem[Aleskerov and Kurbanov (1999)]{AlKurb}Aleskerov, F. and Kurbanov, E. (1999) Degree of manipulability of social choice procedures. \textit{Current Trends in Economics} 8:13--27.

\bibitem[Chamberlin (1985)]{Chamberlin}Chamberlin, J. (1985) An investigation into the relative manipulability of four voting systems. \textit{Behavioral Science} 30(4):195--203.

\bibitem[Cervone et al. (2005)]{CGZ}Cervone, D.P., Gehrlein, W.V. and Zwicker, W.S. (2005) What scoring rule maximizes Condorcet efficiency under IAC. 
\textit{Theory and Decision} 58:145--185.

\bibitem[Devroye (1986)]{Devroye} Devroye, L. (1986)
	Nonuniform random variate generation,
	\textit{Springer-Verlag, New York}


\bibitem[Diss (2015)]{D15}Diss, M. (2015) Strategic manipulability of self--selective social choice rules.  \textit{Annals of Operations Research}  229:347--376.

\bibitem[Diss and Merlin (2020)]{DM20}Diss, M. and  Merlin, V. (2020)  Evaluating voting systems with probability models, essays by and in honor of William Gehrlein and Dominique Lepelley.  Studies in Choice and Welfare -  Springer.

\bibitem[Favardin and Lepelley (2006)]{FavardinLepelley}Favardin, P. and  Lepelley, D. (2006) Some further results on the manipulability of social choice rules. \textit{Social Choice and Welfare}  26:485--509.

\bibitem[Favardin et al. (2002)]{Favardin}Favardin, P., Lepelley, D. and Serais, J. (2002) Borda rule, Copeland method and strategic manipulation. \textit{Review of Economic Design} 7(2):213--228.

\bibitem[Fishburn and Gehrlein (1976)]{GF}Fishburn, P.C. and Gehrlein, W.V. (1976) Condorcet's paradox and anonymous preference profiles. \textit{Public Choice} 26:1--18.

\bibitem[Gehrlein and  Lepelley (2011)]{GehrleinLepelley1}Gehrlein, W.V. and  Lepelley, D. (2011)  \textit{Voting paradoxes and group coherence}. Studies in Choice and Welfare - Springer.

\bibitem[Gehrlein and  Lepelley (2017)]{GehrleinLepelley2}Gehrlein, W.V. and  Lepelley, D. (2017)  \textit{Elections, Voting Rules and Paradoxical Outcomes}. Studies in Choice and Welfare - Springer.

\bibitem[Gehrlein et al. (2013)]{GehrleinMoyouwouLepelley}Gehrlein, W.V., Moyouwou, I. and  Lepelley, D. (2013)  The impact of voters' preference diversity on the probability of some electoral outcomes.  \textit{Mathematical Social Sciences} 66:352--365.

\bibitem[Gibbard (1973)]{Gibbard}Gibbard, A. (1973) Manipulation of voting schemes: a general result. \textit{Econometrica} 41:587--601.

\bibitem[El Ouafdi et al. (2020a)]{Ouafdi1}El Ouafdi, A., Lepelley, D. and   Smaoui, H. (2020a) Probabilities of electoral outcomes: from three-alternative to four-alternative elections.  \textit{Theory and Decision}  88:205-229.

\bibitem[El Ouafdi et al. (2020b)]{Ouafdi2}El Ouafdi, A., Lepelley, D.,   Smaoui, H. and Serais, J. (2020b) Sur la manipulabilit\'e coalitionnelle du vote par note  \`a   trois niveaux. \textit{Working paper}.

\bibitem[Kamwa and Moyouwou (2020)]{KaM}Kamwa, E. and  Moyouwou, I. (2020) Susceptibility  to  manipulation by  sincere truncation: the case of scoring rules and  scoring runoff systems. In \textit{Evaluating Voting Systems with Probability Models, Essays by and in honor of William Gehrlein and Dominique Lepelley}. Diss, M.  and Merlin, V. (editors). Springer, Berlin.

\bibitem[Kelly (1993)]{Kelly}Kelly, J.S. (1993) Almost all social choice rules are highly manipulable, but a few aren't.  \textit{Social Choice and Welfare} 10:161--175.

\bibitem[Kim  and Roush (1996)]{KimRoush}Kim, K.H. and Roush, F.W. (1996) Statistical manipulability of social choice functions. \textit{Group Decision and Negotiation} 5:263--282.

\bibitem[Lepelley et al. (2008)]{Lepelley2}Lepelley, D., Louichi, A. and Smaoui, H. (2008) On Ehrhart polynomials and probability calculations in voting theory. \textit{Social Choice and Welfare} 30(3):363--383. 

\bibitem[Lepelley and Mbih (1987)]{LM87}Lepelley, D. and Mbih, B. (1987) The proportion of coalitionally unstable situations under the Plurality rule. \textit{Economics Letters} 24(4):311--315.

\bibitem[Lepelley and Mbih (1994)]{Lepelley}Lepelley, D. and  Mbih, B. (1994) The vulnerability of four social choice functions to coalitional manipulation of preferences. \textit{Social Choice and Welfare} 11(3):253--265.

\bibitem[Lepelley and Valognes (2003)]{Lepelley3}Lepelley, D. and Valognes, F. (2003) Voting rules, manipulability and social homogeneity. \textit{Public Choice} 116(1/2):165--184. 

\bibitem[Moyouwou and Tchantcho (2017)]{MoyouwouT2017}Moyouwou, I. and   Tchantcho, H. (2017)  Asymptotic vulnerability of positional voting rules to coalitional manipulation. \textit{Mathematical Social Sciences} 89:70--82.

\bibitem[Nitzan (1985)]{Nitzan}Nitzan, S. (1985) The vulnerability of point-voting schemes to preference variation and strategic manipulation, \textit{Public Choice} 47:349--370.

\bibitem[Peleg (1979)]{Peleg}Peleg, B. (1979) A note on manipulability of large voting schemes. \textit{Theory and Decision}  11:401--412.

\bibitem[Pritchard and Wilson (2007a)]{Wilson1}Pritchard, G. and Wilson, M. (2007a) Exact results on manipulability of positional voting rules. \textit{Social Choice and Welfare}  29(3):487--513.

\bibitem[Pritchard and Wilson (2007b)]{Wilson2}Pritchard, G. and Wilson, M. (2007b) Probability calculations under the IAC hypothesis. \textit{Mathematical Social Sciences}  54(3):244--256.

\bibitem[Saari (1990)]{Saari90}Saari, D. (1990)  Susceptibility to manipulation. \textit{Public Choice} 64:21--41.

\bibitem[Satterthwaite (1975)]{Satterthwaite}Satterthwaite, M.A. (1975) Strategy-proofness and Arrow's conditions: Existence and correspondence theorems for voting procedures and social welfare functions. \textit{Journal of Economic Theory} 10(2):187--217.

\bibitem[Sch\"{u}rmann (2013)]{Schurmann}Sch\"{u}rmann, A. (2013) Exploiting polyhedral symmetries in social choice. \textit{Social Choice and Welfare} 40(4):1097--1110.

\end{thebibliography}

\end{document}